\documentclass[letterpaper,10 pt,conference,usenames,dvipsnames]{ieeeconf}  

\IEEEoverridecommandlockouts                              
\usepackage{amsthm} 
\usepackage{amssymb}  

\newtheorem{theorem}{Theorem}
\newtheorem{corollary}{Corollary}
\newtheorem{lemma}{Lemma}
\newtheorem{remark}{Remark}
\newtheorem{definition}{Definition}
\newtheorem{proposition}{Proposition}
\newtheorem{example}{Example}

\usepackage{amsmath,amsfonts,amssymb,color}
\usepackage{mathtools}
\DeclarePairedDelimiter{\floor}{\lfloor}{\rfloor}
\usepackage{epsfig}
\usepackage{algorithm}
\usepackage{algorithmic}
\usepackage{multirow}

\usepackage{array}
\usepackage{multirow}
\usepackage{epstopdf}
\usepackage{tikz}
\usepackage{relsize}
\usetikzlibrary{shapes,arrows}
\usepackage{cite}
\usepackage{pmat}
\usepackage{epstopdf}
\usepackage{tikz}
\usetikzlibrary{shapes}
\definecolor{winered}{rgb}{0.5,0,0}
\usepackage{scalerel}
\usepackage{xcolor}
\usepackage[colorlinks]{hyperref}
\AtBeginDocument{%
  \hypersetup{
    citecolor={},
    linkcolor=blue!}}
\DeclareMathOperator*{\argmax}{\arg\!\max}

\title{\LARGE \bf A Communication-Efficient Algorithm for Exponentially Fast Non-Bayesian Learning in Networks}
\author{Aritra Mitra, John A. Richards, and Shreyas Sundaram
\thanks{A. Mitra, and S. Sundaram are with the School of Electrical and Computer Engineering at Purdue University. J. A.  Richards is with Sandia National Laboratories.   Email: {\tt \{mitra14, sundara2\}@purdue.edu},  {\tt{jaricha@sandia.gov}}. This work was supported in part by NSF CAREER award
1653648, and by a grant from Sandia National Laboratories. Sandia National Laboratories is a multimission laboratory managed and operated by National Technology \& Engineering Solutions of Sandia, LLC, a wholly owned subsidiary of Honeywell International Inc., for the U.S. Department of Energy's National Nuclear Security Administration under contract DE-NA0003525. The views expressed in the article do not necessarily represent the views of the U.S. Department of Energy or the United States Government.}}
\begin{document}
\maketitle
\thispagestyle{empty}
\pagestyle{empty}
\begin{abstract}
    We introduce a simple time-triggered protocol to achieve communication-efficient  non-Bayesian learning over a network. Specifically, we consider a scenario where a group of agents interact over a graph with the aim of discerning the true state of the world that generates their joint observation profiles. To address this problem, we propose a novel distributed learning rule wherein agents aggregate neighboring beliefs based on a min-protocol, and the inter-communication intervals grow geometrically at a rate $a \geq 1$. Despite such sparse communication, we show that each agent is still able to rule out every false hypothesis exponentially fast with probability $1$, as long as $a$ is finite. For the special case when communication occurs at every time-step, i.e., when $a=1$, we prove that the asymptotic learning rates resulting from our algorithm are network-structure independent, and a strict improvement upon those existing in the literature. In contrast, when $a>1$, our analysis reveals that the asymptotic learning rates vary across agents, and exhibit a non-trivial dependence on the network topology coupled with the relative entropies of the agents' likelihood models. This motivates us to consider the problem of allocating signal structures to agents to maximize appropriate performance metrics. In certain special cases, we show that the eccentricity centrality and the decay centrality of the underlying graph help identify optimal allocations; for more general scenarios, we bound the deviation from the optimal allocation as a function of the parameter $a$, and the diameter of the communication graph.
\end{abstract}
\section{Introduction}
A typical problem in networked systems involves a global task that needs to be accomplished by a group of entities or agents via local computations and information exchanges over the network. These agents, however, are typically endowed with partial information about the state of the system; as such, inter-agent communication becomes indispensable for achieving the common goal. Given this premise, it is natural to ask: how frequently must the agents communicate to solve the desired problem? Owing to its practical relevance, the question posed above has received significant recent interest by the control system, information theory and machine learning communities in the context of a variety of problems, namely average consensus \cite{olshevsky}, optimization \cite{opt1,opt2,opt3}, and static parameter estimation \cite{sahu}. Our goal in this paper is to extend such investigations to the problem of non-Bayesian learning in a network, also known as the distributed hypothesis testing problem \cite{jad1,jad2,shahin,nedic,lalitha,mitraACC19}. Specifically, the global task in this setting involves learning the true state of the world (among a finite set of hypotheses) that explains the private observations of each agent in the network. Two notable features that are specific to this problem are as follows. Unlike consensus or distributed optimization, agents are privy to exogenous signals, which, if informative, can enable them to eliminate a subset of the false hypotheses exponentially fast.
A related problem where agents receive exogenous signals (measurements) is that of distributed state estimation \cite{martins,mitraTAC} where the global task entails tracking potentially unstable dynamics. In contrast, the true state of the world remains fixed over time in our setting, considerably simplifying the objective. These attributes play in favor of the problem at hand, motivating us to ask the following questions. (i) Can we design an algorithm that enables each agent to  learn the truth with sparse communication schedules (and in fact, even sparser than typically employed for other classes of distributed problems)? (ii) If so, how fast do the agents learn the truth? (iii) Can we quantify the trade-off(s) between sparsity in communication and the rate of learning? To the best of our knowledge, these questions remain largely unexplored. In this paper, we take a preliminary step towards responding to them via the following \textbf{contributions}.

We develop and analyze a simple time-triggered learning rule that builds on our recent work on distributed hypothesis testing \cite{mitraACC19}. Specifically, the data-aggregation step of our algorithm involves a min-protocol as opposed to the consensus-based averaging schemes intrinsic to existing linear \cite{jad1,jad2} and log-linear \cite{shahin,nedic,lalitha} learning rules. The basic strategy we employ to achieve communication-efficiency is in line with those proposed in \cite{olshevsky,opt1,sahu}, where inter-agent communications become progressively sparser as time evolves. In particular, the authors in \cite{olshevsky} and \cite{opt1} explore deterministic rules where the inter-communication intervals grow logarithmically and polynomially in time, respectively. In contrast, the authors in \cite{sahu} propose a rule where at each time-step, an agent communicates with its neighbors in the graph with a probability that decays to zero at a sub-linear rate. In essence, these approaches establish that as long as the inter-communication intervals do not grow too fast, the global task can still be achieved. We depart from these approaches by allowing the inter-communication intervals to grow much faster: at a  geometric rate $a\geq1$, where the parameter $a$ can be adjusted to control the frequency of communication. While more refined approaches to achieve communication-efficiency are  conceivable, we show that our simple time-triggered protocol yields strong guarantees. Specifically, we prove that even with an arbitrarily large $a$ (which leads to a highly sparse communication schedule), each agent is still able to learn the truth with probability $1$, provided $a$ is finite. Furthermore, we establish that such learning occurs exponentially fast, and characterize the limiting error exponents as a function of certain parameters of our model, and the constant $a$. In particular, our characterization quantifies the trade-offs between communication-efficiency and the speed of learning for the specific problem under consideration.  

Our analysis subsumes the special case when communication occurs at every time-step, i.e., when $a=1$, which corresponds to the scenario studied in our previous work \cite{mitraACC19}. While the general approach in \cite{mitraACC19} was shown to be robust to worst-case adversarial attack models, a convergence-rate analysis of the same was missing. A significant contribution of this paper is to fill this gap by establishing that when $a=1$, \textit{the asymptotic learning rates resulting from our proposed algorithm are network-structure independent, and a strict improvement over the rates provided by  existing algorithms in the literature}. In contrast, when $a>1$, we show that the asymptotic learning rates differ from agent to agent, and depend not only on the relative entropies of the agents' signal models, but also on properties of the underlying network. Given this result, we introduce two new measures of the quality of learning, and study the problem of allocating signal structures to agents to maximize such measures. In certain special cases, we show that the eccentricity centrality and the decay centrality of the communication network play key roles in identifying the optimal allocations. For more general cases, we bound the deviation from the optimal allocation as a function of the parameter $a$, and the diameter of the graph.

\section{Model and Problem Formulation}
\label{sec:model}
\textbf{Network Model:} We consider a setting comprising of a group of agents $\mathcal{V}=\{1,2,\ldots,n\}$. At certain specific time-steps (to be decided by a time-triggered communication schedule), these agents
interact with each other over a directed graph $\mathcal{G}=(\mathcal{V},\mathcal{E})$. An edge $(i,j)\in\mathcal{E}$ indicates that agent $i$ can directly transmit information to agent $j$; in such a case, agent $i$ will be called a neighbor of agent $j$. The set of all neighbors of agent $i$ will be denoted $\mathcal{N}_i$. For a strongly-connected graph $\mathcal{G}$, we will use $d(i,j)$ to denote the length of the shortest directed path from agent $i$ to agent $j$, and $\bar{d}({\mathcal{G}})$ to denote the diameter of the graph.\footnote{A graph is said to be strongly-connected if it has a directed path between every pair of agents; the diameter of such a graph is the length of the longest shortest path between the agents.}

\textbf{Observation Model:} Let $\Theta=\{\theta_1,\theta_2,\ldots,\theta_m\}$ denote $m$ possible states of the world, with each state representing a hypothesis. A specific state $\theta^{\star}\in\Theta$, referred to as the true state of the world, gets realized. Conditional on its realization, at each time-step $t\in\mathbb{N}_{+}$, every agent $i\in\mathcal{V}$  privately observes a signal $s_{i,t}\in\mathcal{S}_i$, where $\mathcal{S}_i$ denotes the signal space of agent $i$.\footnote{We use $\mathbb{N}$ and $\mathbb{N}_{+}$ to refer to the set of non-negative integers and positive integers, respectively.} The joint observation profile so generated across the network is denoted ${s}_{t}=(s_{1,t},s_{2,t},\ldots,s_{n,t})$, where $s_t\in\mathcal{S}$, and $\mathcal{S}=\mathcal{S}_1\times\mathcal{S}_2\times\ldots \mathcal{S}_n$. 
Specifically, the signal $s_{t}$ is generated based on a conditional likelihood function $l(\cdot|\theta^{\star})$, the $i$-th marginal of which is denoted $l_i(\cdot|\theta^{\star})$, and is available to agent $i$. The signal structure of each agent $i\in\mathcal{V}$ is thus characterized by a family of parameterized marginals $l_i=\{l_i(w_i|\theta): \theta\in\Theta, w_i\in\mathcal{S}_i\}$. We make certain standard assumptions \cite{jad1,jad2,shahin,nedic,lalitha}: (i) The signal space of each agent $i$, namely $\mathcal{S}_i$, is finite. (ii) Each agent $i$ has knowledge of its local likelihood functions $\{l_i(\cdot|\theta_p)\}_{p=1}^{m}$, and it holds that $l_i(w_i|\theta) > 0, \forall w_i\in\mathcal{S}_i$, and $\forall \theta \in \Theta$. (iii) The observation sequence of each agent is described by an i.i.d. random process over time; however, at any given time-step, the observations of different agents may potentially be correlated. (iv) There exists a fixed true state of the world $\theta^{\star}\in\Theta$ (that is unknown to the agents) that generates the observations of all the agents. The probability space for our model is denoted $(\Omega,\mathcal{F},\mathbb{P}^{\theta^{\star}})$, where $\Omega\triangleq\{\omega: \omega=(s_1,s_2,\ldots), \forall s_t\in\mathcal{S}, \forall t \in \mathbb{N}_{+}\}$, $\mathcal{F}$ is the $\sigma$-algebra generated by the observation profiles, and $\mathbb{P}^{\theta^{\star}}$ is the probability measure induced by sample paths in $\Omega$. Specifically, $\mathbb{P}^{\theta^{\star}}=\prod \limits_{t=1}^{\infty}l(\cdot|\theta^{\star})$. We will use the abbreviation a.s. to indicate almost sure occurrence of an event w.r.t. $\mathbb{P}^{\theta^{\star}}$.

Given the above setting, the goal of each agent in the network is to eventually learn the true state of the world $\theta^{\star}$. However, the signal structure of any given agent is in general only partially informative, thereby precluding this task from being achieved by any agent in isolation. Specifically, let $\Theta^{\theta^{\star}}_i\triangleq\{\theta\in\Theta : l_i(w_i|\theta)=l_i(w_i|\theta^{\star}), \forall w_i\in\mathcal{S}_i\}$  represent the set of hypotheses that are \textit{observationally equivalent} to the true state $\theta^{\star}$ from the perspective of agent $i$. An agent $i$ is deemed partially informative about the truth if  $|\Theta^{\theta^{\star}}_i| > 1$. Since potentially every agent can be partially informative in the sense described above, inter-agent communications become necessary for each agent to learn the truth. 

In this context, our \textbf{objectives} in this paper are to develop an understanding of (i) the amount of leeway that the above problem affords in terms of sparsifying inter-agent communications without compromising the objective of learning the truth, and (ii) the trade-offs between sparse communication and the rate of learning. To this end, we recall the following definition from \cite{mitraACC19} that will prove useful in our subsequent developments. 

\begin{definition}(\textbf{Source agents}) An agent $i$ is said to be a source agent for a pair of distinct hypotheses $\theta_p,\theta_q\in\Theta$ if it can distinguish between them, i.e., if  $D(l_i(\cdot|\theta_p)||l_i(\cdot|\theta_q)) > 0$, where $D(l_i(\cdot|\theta_p)||l_i(\cdot|\theta_q))$ represents the KL-divergence \cite{cover} between the distributions $l_i(\cdot|\theta_p)$ and $l_i(\cdot|\theta_q)$. The set of source agents for pair $(\theta_p,\theta_q)$ is denoted $\mathcal{S}(\theta_p,\theta_q)$.
\end{definition}

Throughout the rest of the paper, we will use $K_i(\theta_p,\theta_q)$
as a shorthand for $D(l_i(\cdot|\theta_p)||l_i(\cdot|\theta_q))$.

\section{A Communication-Efficient Learning Rule}
\label{sec:Algo}
In this section, we formally introduce a simple time-triggered belief update rule parameterized by a constant $a\in\mathbb{N}_{+}$ that determines the frequency of communication (to be made more precise below). In order to collaboratively learn the true state of the world, every agent $i$ maintains a local belief vector $\boldsymbol{\pi}_{i,t}$, and an actual belief vector $\boldsymbol{\mu}_{i,t}$, each of which are probability distributions over the hypothesis set $\Theta$. These vectors are initialized with $\pi_{i,0}(\theta)>0,\mu_{i,0}(\theta)>0, \forall \theta\in\Theta,\forall i\in\mathcal{V}$ (but otherwise arbitrarily), and subsequently updated as follows. 
\begin{itemize}
\item \underline{\textbf{Update of the local beliefs}:} At each time-step $t+1\in\mathbb{N}_{+}$, the local belief vectors are updated based on a standard Bayesian rule:
\begin{equation}
\pi_{i,t+1}(\theta)=\frac{l_i(s_{i,t+1}|\theta)\pi_{i,t}(\theta)}{\sum  \limits_{p=1}^{m} l_i(s_{i,t+1}|\theta_p)\pi_{i,t}(\theta_p)}.
\label{eqn:Bayes}
\end{equation}
\item \underline{\textbf{Update of the actual beliefs}:}
Let $\mathbb{I}={\{t_k\}}_{ k\in\mathbb{N}_{+}}$ denote a sequence of time-steps satisfying $t_{k+1}-t_k=a^k, \forall k\in\mathbb{N}_{+}$, with $t_1=1$. If $t+1\in\mathbb{I}$, then $\boldsymbol{\mu}_{i,t+1}$ is updated as
\begin{equation}
\mu_{i,t+1}(\theta)=\frac{\min\{\{\mu_{j,t}(\theta)\}_{{j\in\mathcal{N}_i}},\pi_{i,t+1}(\theta)\}}{\sum\limits_{p=1}^{m}\min\{\{\mu_{j,t}(\theta_p)\}_{{j\in\mathcal{N}_i}},\pi_{i,t+1}(\theta_p)\}}.
\label{eqn:update1}
\end{equation}
If $t+1\notin\mathbb{I}$,  $\boldsymbol{\mu}_{i,t+1}$ is simply held constant as follows:
\begin{equation}
    \mu_{i,t+1}(\theta)=\mu_{i,t}(\theta).
    \label{eqn:update2}
\end{equation}
\end{itemize}
In words, while the local beliefs are updated at every time-step, the actual beliefs are updated only at time-steps that belong to the set $\mathbb{I}$, i.e., an agent $i\in\mathcal{V}$ is allowed to transmit  $\boldsymbol{\mu}_{i,t}$ to its out-neighbors, and receive $\boldsymbol{\mu}_{j,t}$ from each in-neighbor $j$ in $\mathcal{G}$ if and only if $t+1\in\mathbb{I}$. When $a=1$, the actual beliefs get updated as per \eqref{eqn:update1} at \textit{every} time-step, and we recover the rule proposed in  \cite{mitraACC19}. When $a>1$, note that the inter-communication intervals grow exponentially at a rate dictated by the parameter $a$. Our goal in this paper is to precisely characterize the impact of such sparse communication on the asymptotic rate of learning of each agent. Prior to doing so, a few comments are in order.

First, notice that the data-aggregation rule in \eqref{eqn:update1} is based on a min-protocol, as opposed to any form of ``belief-averaging" commonly employed in the existing distributed learning literature \cite{jad1,jad2,shahin,nedic,lalitha}. Essentially, while the local belief updates \eqref{eqn:Bayes} capture what an agent can learn by itself, the actual belief updates \eqref{eqn:update1} incorporate information from the rest of the network. As demonstrated by Corollary \ref{thm:corollary} in the next section, when $a=1$, such a min-protocol yields better asymptotic learning rates than all existing schemes. This motivates us to use a belief update rule of the form \eqref{eqn:update1} for studying the case when $a>1$. 
Second, we note that the proposed time-triggered protocol is simple, easy to implement and computationally cheap. At the same time, the exponentially growing intervals afford a much sparser communication schedule relative to related literature. Third, while one can potentially consider extensions of this algorithm that account for asynchronicity, communication failures, delays etc., we focus on the scheme here in order to (i) concretely isolate the trade-off between sparse communication and the quality of learning as measured by the asymptotic learning rates, and (ii) provide insights into how the network structure impacts such rates. A final comment needs to made regarding the choice of achieving communication-efficiency by cutting down on communication rounds as opposed to truncating the number of bits exchanged per communication round, an approach pursued in quantization-based schemes \cite{suresh}. As argued in \cite{opt2}, communication latency acts as the bottleneck of overall performance and dominates message-size dependent transmission latency when it comes to transmitting small messages, such as the $m$-dimensional actual belief vectors in our setting. This justifies our sparse communication scheme.
With these points in mind, we proceed to the analysis of the algorithm developed in this section. 
\section{Main Result and Discussion}
The main result of the paper is as follows; the proof of this result is presented in Section \ref{sec:proofs}. 
\begin{theorem}
Suppose the communication parameter satisfies $a>1$, and the following conditions are met.
\begin{enumerate}
\item[(i)] For every pair of hypotheses $\theta_p,\theta_q\in\Theta$, the corresponding source set $\mathcal{S}(\theta_p,\theta_q)$ is non-empty.
\item[(ii)] The communication graph $\mathcal{G}$ is strongly-connected.
\item[(iii)] Every agent $i\in\mathcal{V}$ has a non-zero prior belief on each hypothesis, i.e., $\pi_{i,0}(\theta) > 0,\mu_{i,0}(\theta) > 0$ for all $i\in\mathcal{V}$, and for all $\theta\in\Theta$.
\end{enumerate}
Then, the time-triggered distributed learning rule described by equations \eqref{eqn:Bayes}, \eqref{eqn:update1}, \eqref{eqn:update2} provides the following guarantees.
\begin{itemize}
    \item \textbf{(Consistency)}: For each agent  $i\in\mathcal{V}$,  $\mu_{i,t}(\theta^{\star}) \rightarrow 1$ a.s.
    \item \textbf{(Asymptotic Rate of Rejection of False Hypotheses)}: For each agent  $i\in\mathcal{V}$, and for each false hypothesis $\theta\in\Theta\setminus\{\theta^{\star}\},$ the following holds:
    \begin{equation}
        \liminf_{t\to\infty}-\frac{\log\mu_{i,t}(\theta)}{t} \geq \max_{v\in\mathcal{S}(\theta^{\star},\theta)}\frac{K_v(\theta^{\star},\theta)}{a^{(d(v,i)+1)}} \hspace{1mm} a.s.
        \label{eqn:asymprate}
    \end{equation}
\end{itemize}
\label{thm:main}
\end{theorem}

We obtain the following important corollary, the proof of which follows readily from that of Theorem \ref{thm:main} in Section \ref{sec:proofs}.
\begin{corollary}
Suppose communication occurs at every time-step, i.e., suppose $a=1$. Let conditions (i)-(iii) in the statement of Theorem \ref{thm:main} hold. Then, our proposed learning rule guarantees consistency in the same sense as in Theorem \ref{thm:main}. Furthermore, for each agent $i\in\mathcal{V}$, and for each false hypothesis $\theta\in\Theta\setminus\{\theta^{\star}\}$, the following holds:
    \begin{equation}
        \liminf_{t\to\infty}-\frac{\log\mu_{i,t}(\theta)}{t} \geq \max_{v\in\mathcal{S}(\theta^{\star},\theta)}K_v(\theta^{\star},\theta) \hspace{1mm} a.s.
        \label{eqn:asymprate2}
    \end{equation}
\label{thm:corollary}
\end{corollary}

We remark on the implications of the above results.

 \textbf{Implications of Theorem \ref{thm:main}}: We first note that despite its simplicity, the time-triggered algorithm proposed in Section \ref{sec:Algo} provides strong guarantees: Eqn. \eqref{eqn:asymprate} indicates that although the inter-communication intervals grow exponentially at an arbitrarily large (but finite) rate $a$, each agent is still able to  eliminate every false hypothesis at an exponential rate with probability $1$.  More interestingly, \eqref{eqn:asymprate} reveals that the asymptotic learning rates are \textit{agent-specific}, i.e., different agents may discover the truth at different rates.\footnote{We use the lower bounds derived in \eqref{eqn:asymprate}, \eqref{eqn:asymprate2} as a proxy when referring to the corresponding asymptotic learning rates.} In particular, when considering the asymptotic rate of rejection of a particular false hypothesis at a given agent $i$, notice from the RHS of \eqref{eqn:asymprate} that one needs to account for the attenuated relative entropies of the corresponding source agents, where the attenuation factor scales exponentially with the distances of agent $i$ from such source agents. This contrasts with existing literature \cite{jad1,jad2,shahin,lalitha,nedic}, and the case when $a=1$ in Corollary \ref{thm:corollary}, where all agents learn the truth at identical rates. 

 \textbf{Implications of Corollary \ref{thm:corollary}}: In sharp contrast to the case when $a>1$, Corollary \ref{thm:corollary} indicates that when communication occurs at every time-step (i.e., $a=1$), the asymptotic learning rates are \textit{network-structure independent}, and \textit{identical} for each agent. Since this case represents the standard distributed hypothesis testing setup studied in literature, it becomes important to know how such rates compare with those resulting from existing  ``belief-averaging" schemes \cite{jad1,jad2,shahin,lalitha,nedic}. To this end, we note that under the same set of assumptions as in Theorem \ref{thm:main}, both linear \cite{jad1,jad2} and log-linear \cite{shahin,lalitha,nedic} opinion pooling lead to an asymptotic rate of rejection of the form $\sum_{i\in\mathcal{V}}\nu_iK_i(\theta^{\star},\theta)$ for each false hypothesis $\theta\in\Theta\setminus\{\theta^{\star}\}$, and the rate is identical for each agent. Here, $\nu_i$ represents the eigenvector centrality of agent $i\in\mathcal{V}$. It is well known that for a strongly-connected graph, $\nu_i>0, \forall i\in \mathcal{V}$. Thus, based on the above discussion, and referring to \eqref{eqn:asymprate2}, we conclude that a significant contribution of the algorithm proposed in this paper is that it yields \textit{strictly better} asymptotic learning rates than those existing in the literature, for the standard setting when $a=1$.\footnote{Recently, in \cite{mitraTAC19}, we showed that this result continues to hold even if the underlying graph changes with time, but satisfies a mild joint-strong connectivity condition.}

 \textbf{Trade-Off between Sparse Communication and Quality of Learning}: From \eqref{eqn:asymprate}, it is apparent that sparser communication schedules (corresponding to larger $a$'s) come at the cost of lower asymptotic learning rates. Furthermore, since such rates depend upon the network-structure when $a>1$, a poor allocation of signal structures to agents can have adverse effects on the learning rates of certain agents. However, the above problem is readily bypassed when $a=1$, since the learning rates for that case solely depend on the relative entropies of the agents, as shown by \eqref{eqn:asymprate2}. 

\section{Proof of the Main Result}
\label{sec:proofs}
In order to prove Theorem \ref{thm:main}, we require a few intermediate results. The first one is a standard consequence of Bayesian updating,  and characterizes the behavior of the local belief trajectories generated via \eqref{eqn:Bayes}; for a proof, see \cite{mitraACC19}. 
\begin{lemma}
Consider a false hypothesis $\theta\in\Theta\setminus\{\theta^{\star}\}$, and an
agent $i\in\mathcal{S}(\theta^{\star},\theta)$. Suppose $\pi_{i,0}(\theta_p) > 0, \forall \theta_p\in\Theta$. Then, the update rule \eqref{eqn:Bayes} ensures that (i) $\pi_{i,t}(\theta) \rightarrow 0$ a.s., (ii) $\pi_{i,\infty}(\theta^{\star})\triangleq\lim_{t\to\infty}\pi_{i,t}(\theta^{\star})$ exists a.s. and satisfies $\pi_{i,\infty}(\theta^{\star})\geq \pi_{i,0}(\theta^{\star})$, and (iii) the following holds:
\begin{equation}
\lim_{t\to\infty}\frac{1}{t}\log\frac{\pi_{i,t}(\theta)}{\pi_{i,t}(\theta^{\star})}=-K_i(\theta^{\star},\theta) \hspace{1mm} a.s.
\label{eqn:localrate}
\end{equation}
\label{lemma:Bayes}
\end{lemma}

\begin{lemma}
 Suppose the conditions in Theorem \ref{thm:main} hold, and the learning rule given by \eqref{eqn:Bayes}, \eqref{eqn:update1}, and \eqref{eqn:update2} is employed by each agent. Then, there exists a set $\bar{\Omega}\subseteq\Omega$ with the following properties: (i) $\mathbb{P}^{\theta^{\star}}(\bar{\Omega})=1$, and (ii) for each $\omega\in\bar{\Omega}$, there exist constants $\eta(\omega)\in(0,1)$ and $t'(\omega)\in(0,\infty)$ such that
\begin{equation}
\pi_{i,t}(\theta^{\star}) \geq \eta(\omega), \mu_{i,t}(\theta^{\star}) \geq \eta(\omega), \forall t \geq t'(\omega),\forall i\in\mathcal{V}.
\label{eqn:lowerbound}
    \end{equation}
\label{lemma:bound}
\end{lemma}
\begin{proof} Let $\bar{\Omega}\subseteq\Omega$ denote the set of sample paths for which the assertions in Lemma \ref{lemma:Bayes} hold for each false hypothesis $\theta\in\Theta\setminus\{\theta^{\star}\}$. Based on Lemma \ref{lemma:Bayes}, we note that $\mathbb{P}^{\theta^{\star}}(\bar{\Omega})=1$. Consequently, to prove the result, it suffices to establish the existence of $\eta(\omega) \in (0,1)$, and $t'(\omega)\in (0,\infty)$ such that \eqref{eqn:lowerbound} holds for each sample path $\omega \in \bar{\Omega}$. To this end, pick an arbitrary sample path $\omega\in\bar{\Omega}.$ We first argue that the local beliefs of every agent on the true state $\theta^{\star}$ are bounded away from $0$ on $\omega$. To see this, pick any agent $i\in\mathcal{V}$. Suppose there exists some $\theta\in\Theta\setminus\{\theta^{\star}\}$ for which $i\in\mathcal{S}(\theta^{\star},\theta)$. Then, based on our choice of $\omega$, it follows directly from Lemma \ref{lemma:Bayes} that $\pi_{i,\infty}(\theta^{\star})\geq\pi_{i,0}(\theta^{\star})>0$, where the last inequality follows from condition (iii) in Theorem \ref{thm:main}. In particular, given the structure of the update rule \eqref{eqn:Bayes}, it follows that $\pi_{i,t}(\theta^{\star}) > 0$ for all time (since if $\pi_{i,t}(\theta^{\star})=0$ at any instant, then the corresponding belief would remain at $0$ for all subsequent time-steps, thereby violating the fact that $\pi_{i,\infty}(\theta^{\star})\geq\pi_{i,0}(\theta^{\star})>0$). If there exists no  $\theta\in\Theta\setminus\{\theta^{\star}\}$ for which $i\in\mathcal{S}(\theta^{\star},\theta)$, then every hypothesis in $\Theta$ is observationally equivalent to $\theta^{\star}$ from the point of view of agent $i$. In this case, it is easy to see that based on \eqref{eqn:Bayes}, $\boldsymbol{\pi}_{i,t}=\boldsymbol{\pi}_{i,0}, \forall t\in\mathbb{N}_{+}$. In particular, this implies $\pi_{i,t}(\theta^{\star})=\pi_{i,0}(\theta^{\star})>0, \forall t \in \mathbb{N}_{+}$. This establishes our claim that on $\omega$, the local beliefs of all the agents remain bounded away from $0$. 

 To proceed, define $\gamma_1\triangleq\min_{i\in\mathcal{V}} \pi_{i,0}(\theta^{\star})>0$, where the inequality follows from condition (iii) in Theorem \ref{thm:main}. Pick a small number $\delta > 0$ such that $\delta < \gamma_1$, and notice that our discussion concerning the evolution of the local beliefs readily implies the existence of a time-step $t'(\omega)$, such that for all $t \geq t'(\omega)$, $\pi_{i,t}(\theta^{\star}) \geq  \gamma_1-\delta > 0, \forall i\in\mathcal{V}$. Now define $\gamma_2(\omega)\triangleq\min_{i\in\mathcal{V}}\{\mu_{i,t'(\omega)}(\theta^{\star})\}$, and observe that $\gamma_2(\omega) > 0$. This observation follows from the fact that given the structure of the update rules \eqref{eqn:update1} and \eqref{eqn:update2}, and condition (iii) in Theorem \ref{thm:main},  $\gamma_2(\omega)$ can equal $0$
if and only if some agent in the network sets its local belief on $\theta^{\star}$ to $0$ at some time-step prior to $t'(\omega)$. However, this possibility is ruled out in view of the previously established fact that on $\omega$, $\pi_{i,t}(\theta^{\star})>0, \forall t\in\mathbb{N}, \forall i\in\mathcal{V}$.
Let $\eta(\omega)=\min\{\gamma_1-\delta,\gamma_2(\omega)\} > 0$. It is apparent from the preceding discussion that $\pi_{i,t}(\theta^{\star})\geq\eta(\omega),\forall t\geq t'(\omega),\forall i\in\mathcal{V}$. It remains to establish a similar result for the actual beliefs $\mu_{i,t}(\theta^{\star})$.
To this end, let $\bar{t}(\omega) > t'(\omega)$ be the first time-step following $t'(\omega)$ that belongs to the set $\mathbb{I}$. Based on \eqref{eqn:update2}, notice that $\mu_{i,t}(\theta^{\star}) \geq \eta(\omega)$ for all $t\in[t'(\omega),\bar{t}(\omega))$, and for each $i\in\mathcal{V}$. Based on \eqref{eqn:update1}, at time-step $\bar{t}(\omega)\in\mathbb{I}$, $\mu_{i,\bar{t}(\omega)}(\theta^{\star})$ for an agent $i\in\mathcal{V}$ satisfies:
\begin{equation}
\resizebox{0.85\hsize}{!}{$
\begin{aligned}
\mu_{i,\bar{t}(\omega)}(\theta^{\star})&\geq\frac{\eta(\omega)}{\sum\limits_{p=1}^{m}\min\{\{\mu_{j,\bar{t}(\omega)-1}(\theta_p)\}_{{j\in\mathcal{N}_i}},\pi_{i,\bar{t}(\omega)}(\theta_p)\}}\\
&\geq\frac{\eta(\omega)}{\sum\limits_{p=1}^{m}\pi_{i,\bar{t}(\omega)}(\theta_p)}=\eta(\omega),
\end{aligned}
$}
\end{equation}
where the last equality follows from the fact that the local belief vectors generated via \eqref{eqn:Bayes} are valid probability distributions over the hypothesis set $\Theta$ at each time-step, and hence $\sum\limits_{p=1}^{m}\pi_{i,\bar{t}(\omega)}(\theta_p)=1$. The above argument applies identically to each agent in $\mathcal{V}$. Furthermore, it is easily seen that based on \eqref{eqn:update2}, and a similar reasoning as above, identical conclusions can be drawn for each time-step $t>t'(\omega),t\in\mathbb{I}$ when the agents update their actual beliefs based on \eqref{eqn:update1}. This readily establishes \eqref{eqn:lowerbound}, and completes the proof. 
\end{proof}
\begin{lemma}
Consider a false hypothesis $\theta\in\Theta\setminus\{\theta^{\star}\}$ and an agent $v\in\mathcal{S}(\theta^{\star},\theta)$. Suppose the conditions stated in Theorem \ref{thm:main} hold. Then, the learning rule described by equations \eqref{eqn:Bayes}, \eqref{eqn:update1} and \eqref{eqn:update2} guarantee the following for each agent $i\in\mathcal{V}$:
\begin{equation}
        \liminf_{t\to\infty}-\frac{\log\mu_{i,t}(\theta)}{t} \geq \frac{K_v(\theta^{\star},\theta)}{a^{(d(v,i)+1)}} \hspace{1mm} a.s.
        \label{eqn:eachsource}
    \end{equation}
\label{lemma:main}
\end{lemma}
\begin{proof}
Throughout this proof, we use the same notation as in the proof of Lemma \ref{lemma:bound}. With $\bar{\Omega}$ as in Lemma \ref{lemma:bound}, pick an arbitrary sample path $\omega\in\bar{\Omega}$, an agent $v\in\mathcal{S}(\theta^{\star},\theta)$, and an agent $i\in\mathcal{V}$. Since condition (ii) in Theorem \ref{thm:main} is met, there exists a directed path of shortest length from agent $v$ to agent $i$ in $\mathcal{G}$. To prove the result, we shall induct on the length of such a path. First, we consider the base case when $d(v,i)=0$, i.e., when $i=v$. In other words, we will analyze the asymptotic rate of rejection of $\theta$ at the source agent $v$. Fix $\epsilon > 0$, and notice that since $v\in\mathcal{S}(\theta^{\star},\theta)$, Lemma \ref{lemma:Bayes} implies that there exists ${t}_v(\omega,\theta,\epsilon)\in\mathbb{N}_{+}$, such that:
\begin{equation}
    \pi_{v,t}(\theta) <  e^{-(K_v(\theta^{\star},\theta)-\epsilon)t}, \forall t \geq {t}_v(\omega,\theta,\epsilon).
    \label{eqn:bound1}
\end{equation}
Since $\omega\in\bar{\Omega}$,  Lemma \ref{lemma:bound} guarantees the existence of a time-step $t'(\omega) < \infty$, and a constant $\eta(\omega)>0$, such that on $\omega$, $\pi_{i,t}(\theta^{\star}) \geq \eta(\omega), \mu_{i,t}(\theta^{\star}) \geq \eta(\omega), \forall t\geq t'(\omega), \forall i\in\mathcal{V}$. Let $\bar{t}_v(\omega,\theta,\epsilon)=\max\{t'(\omega),t_v(\omega,\theta,\epsilon)\}$. For the remainder of the proof, to simplify the notation,  we suppress the dependence of various quantities on the parameters $\omega,\theta$, and $\epsilon$, since such dependence can be easily inferred from context. Let $\tilde{t} > \bar{t}_v$ be the first time-step following $\bar{t}_v$ that belongs to $\mathbb{I}$, i.e., a time-step when agent $v$ updates its actual beliefs based on \eqref{eqn:update1}. Then, based on the preceding discussion and \eqref{eqn:update1}, we have:
\begin{equation}
\resizebox{0.75\hsize}{!}{$
\begin{aligned}
\mu_{v,\tilde{t}}(\theta)&\overset{(a)}{\leq}\frac{\pi_{v,\tilde{t}}(\theta)}{\sum\limits_{p=1}^{m}\min\{\{\mu_{j,\tilde{t}-1}(\theta_p)\}_{{j\in\mathcal{N}_i}},\pi_{v,\tilde{t}}(\theta_p)\}}\\
&\overset{(b)}{<}\frac{e^{-(K_v(\theta^{\star},\theta)-\epsilon)\tilde{t}}}{\eta(\omega)}
=C(\omega)e^{-(K_v(\theta^{\star},\theta)-\epsilon)\tilde{t}},
\end{aligned}
$}
\label{eqn:upperbnd1}
\end{equation}
where $C(\omega)={\eta(\omega)}^{-1}$. Regarding the inequalities in \eqref{eqn:upperbnd1}, (a) follows directly from \eqref{eqn:update1}, whereas (b) follows from \eqref{eqn:bound1} and the fact that $\eta(\omega)$ lower bounds the beliefs (both local and actual) of all agents on the true state $\theta^{\star}$. Note that consecutive trigger-points $t_{k},t_{k+1}\in\mathbb{I}$ satisfy $t_{k+1}=at_{k}+1$. Based on \eqref{eqn:update2}, we then have:
\begin{equation}
    \mu_{v,t}(\theta) < C(\omega)e^{-(K_v(\theta^{\star},\theta)-\epsilon)\tilde{t}}, \forall t\in[\tilde{t},a\tilde{t}+1 ).
\end{equation}
Based on our rule, the next update of $\mu_{v,t}(\theta)$ takes place at time-step $a\tilde{t}+1$. Employing the same reasoning as we did to arrive at \eqref{eqn:upperbnd1}, we obtain:
\begin{equation}
\mu_{v,a\tilde{t}+1}(\theta)<C(\omega)e^{-(K_v(\theta^{\star},\theta)-\epsilon)(a\tilde{t}+1)}.
\label{eqn:upperbnd2}
\end{equation}
 Coupled with the above inequality, \eqref{eqn:update2} once again implies:
\begin{equation}
 \mu_{v,t}(\theta) < C(\omega)e^{-(K_v(\theta^{\star},\theta)-\epsilon)(a\tilde{t}+1)}, \forall t\in[a\tilde{t}+1,a^2\tilde{t}+a+1).
 \end{equation}
 Generalizing the above reasoning, we obtain:
 \begin{equation}
     \mu_{v,t}(\theta) < C(\omega)e^{-(K_v(\theta^{\star},\theta)-\epsilon)(a^p\tilde{t}+f(p))},
 \end{equation}
 $\forall t\in[a^p\tilde{t}+f(p),a^{(p+1)}\tilde{t}+af(p)+1)$, $p\in\mathbb{N}$, where
 \begin{equation}
     f(p)=\frac{(a^p-1)}{(a-1)}.
 \end{equation}
 This immediately leads to the conclusion that for any $t\geq\tilde{t}$:
 \begin{equation}
     \mu_{v,t}(\theta) < C(\omega)e^{-(K_v(\theta^{\star},\theta)-\epsilon)(a^{p(t)}\tilde{t}+f(p(t)))},
     \label{eqn:uppbound3}
 \end{equation}
 where 
 \begin{equation}
 p(t)=\floor{g(t)}, \hspace{2mm} g(t)=\frac{\log{\frac{(a-1)t+1}{(a-1)\tilde{t}+1}}}{\log a}.
     \label{eqn:constants}
 \end{equation}
 Taking the natural log on both sides of \eqref{eqn:uppbound3}, dividing throughout by $t$, and simplifying, we obtain that $\forall t\geq\tilde{t}$: 
\begin{equation}
\resizebox{1\hsize}{!}{$
\begin{split}
-\frac{\log\mu_{v,t}(\theta)}{t} &> \frac{(K_v(\theta^{\star},\theta)-\epsilon)(a^{p(t)}\tilde{t}+f(p(t)))}{t}-\frac{\log C(\omega)}{t}.
\end{split}
$}
\label{eqn:upperbnd4}
\end{equation}
Let $\alpha_{v}(\theta,\epsilon)=(K_v(\theta^{\star},\theta)-\epsilon)$. Then, taking the limit inferior on both sides of the above inequality yields:
\begin{equation}
\begin{split}
 \liminf_{t\to\infty}-\frac{\log\mu_{v,t}(\theta)}{t} &\geq \alpha_{v}(\theta,\epsilon)\lim_{t\to\infty}\frac{1}{t}\left[a^{p(t)}\tilde{t}+\frac{a^{p(t)}-1}{a-1}\right]\\
 &\geq \frac{\alpha_{v}(\theta,\epsilon)}{a}\lim_{t\to\infty}\frac{1}{t}\left[a^{g(t)}(\tilde{t}+\frac{1}{a-1})\right]\\
 &=\frac{\alpha_v(\theta,\epsilon)}{a},
 \end{split}
 \end{equation}
where the second inequality follows from the fact that $\floor{x}>x-1, \forall x\in\mathbb{R}$, and the final equality results from further simplifications based on \eqref{eqn:constants}.
Finally, note that $\epsilon$ can be made arbitrarily small in the above inequality, and that the above conclusions hold for a generic sample path $\omega \in \bar{\Omega}$, where $\mathbb{P}^{\theta^{\star}}(\bar{\Omega})=1$. This establishes \eqref{eqn:eachsource} for the case when $d(v,i)=0$, and completes the proof of the base case of our induction. To proceed, suppose \eqref{eqn:eachsource} holds for each node $i\in\mathcal{V}$ satisfying $0\leq d(v,i) \leq q$, where $q$ is a non-negative integer satisfying $q\leq \bar{d}(\mathcal{G})-1$ (recall that $\bar{d}(\mathcal{G})$  represents the diameter of the graph $\mathcal{G}$). Let $i\in\mathcal{V}$ be such that $d(v,i)=q+1.$ Thus, there must exist some node $l\in\mathcal{N}_i$ such that $d(v,l)=q.$ The induction hypothesis applies to this node $l$, and hence, we have:
\begin{equation}
        \liminf_{t\to\infty}-\frac{\log\mu_{l,t}(\theta)}{t} \geq \frac{K_v(\theta^{\star},\theta)}{a^{(q+1)}} \hspace{1mm} a.s.
    \end{equation}
Let $\tilde{\Omega}\subseteq\Omega$ be the set of sample paths for which the above inequality holds. With $\bar{\Omega}$ defined as before, notice that $\mathbb{P}^{\theta^{\star}}(\tilde{\Omega}\cap\bar{\Omega})=1$, since $\tilde{\Omega}$ and $\bar{\Omega}$ each have $\mathbb{P}^{\theta^{\star}}$-measure $1$. Pick an arbitrary sample path $\omega\in\tilde{\Omega}\cap\bar{\Omega}$, and notice that based on arguments identical to the base case, on the sample path $\omega$ there exists a time-step $\bar{t}_l$, such that the beliefs of all agents on $\theta^{\star}$ are bounded below by $\eta(\omega)$ following $\bar{t}_l$, and
\begin{equation}
    \mu_{l,t}(\theta) <  e^{-({H}_l(\theta^{\star},\theta)-\epsilon)t}, \forall t \geq \bar{t}_l,
    \label{eqn:boundl1}
\end{equation}
where $\epsilon>0$ is an arbitrary small number and
\begin{equation}
  H_l(\theta^{\star},\theta)=  \frac{K_v(\theta^{\star},\theta)}{a^{(q+1)}}.
\end{equation}
Proceeding as in the base case, let $\tau>\bar{t}_l$ be the first time-step following $\bar{t}_l$ that belongs to the set $\mathbb{I}$. Noting that $l\in\mathcal{N}_i$, using  \eqref{eqn:update1}, \eqref{eqn:boundl1}, and similar arguments as those used to arrive at \eqref{eqn:upperbnd1}, we obtain:
\begin{equation}
\resizebox{0.8\hsize}{!}{$
\begin{aligned}
\mu_{i,\tau}(\theta)&{\leq}\frac{\mu_{l,\tau-1}(\theta)}{\sum\limits_{p=1}^{m}\min\{\{\mu_{j,\tau-1}(\theta_p)\}_{{j\in\mathcal{N}_i}},\pi_{i,\tau}(\theta_p)\}}\\
&{<}\frac{e^{-(H_l(\theta^{\star},\theta)-\epsilon)(\tau-1)}}{\eta(\omega)}=C_l(\omega)e^{-(H_l(\theta^{\star},\theta)-\epsilon)\tau},
\end{aligned}
$}
\label{eqn:boundl2}
\end{equation}
where 
\begin{equation}
C_l(\omega)=\frac{e^{(H_l(\theta^{\star},\theta)-\epsilon)}}{\eta(\omega)}.
\end{equation}
Repeating the above analysis for each time-step of the form $a^p\tau+f(p),p\in\mathbb{N}_{+}$, using \eqref{eqn:update2}, and following similar arguments as in the base case yields that $\forall t\geq\tau$, 
\begin{equation}
     \mu_{i,t}(\theta) < C_l(\omega)e^{-(H_l(\theta^{\star},\theta)-\epsilon)(a^{\bar{p}(t)}\tau+f(\bar{p}(t)))},
     \label{eqn:boundl3}
 \end{equation}
 where 
 \begin{equation}
 \bar{p}(t)=\floor{
     \bar{g}(t)}, \hspace{2mm} \bar{g}(t)=\frac{\log{\frac{(a-1)t+1}{(a-1)\tau+1}}}{\log a}.
 \end{equation}
Notice that the inequality in \eqref{eqn:boundl3} resembles that in \eqref{eqn:uppbound3}. Thus, the remaining steps can be completed identically as the base case to yield:
\begin{equation}
    \liminf_{t\to\infty}-\frac{\log\mu_{i,t}(\theta)}{t} \geq  \frac{H_l(\theta^{\star},\theta)}{a}-\frac{\epsilon}{a}.
\end{equation}
The induction step, and in turn the proof can be completed by substituting the expression for $H_l(\theta^{\star},\theta)$ in the above inequality and recalling that $d(v,i)=q+1$.  
\end{proof}
We are now in position to prove Theorem \ref{thm:main}.
\begin{proof} (\textbf{Theorem \ref{thm:main}}) Fix a  $\theta\in\Theta\setminus\{\theta^{\star}\}$. Based on condition (i) of the Theorem, $\mathcal{S}(\theta^{\star},\theta)$ is non-empty, and based on condition (ii), there exists a path from each agent $v\in\mathcal{S}(\theta^{\star},\theta)$ to every agent in $\mathcal{V}\setminus\{v\}$. Eq.  \eqref{eqn:asymprate} then follows from Lemma \ref{lemma:main}. By definition of a source set, $K_v(\theta^{\star},\theta)>0, \forall v\in\mathcal{S}(\theta^{\star},\theta)$; \eqref{eqn:asymprate} then implies $\lim_{t\to\infty}\mu_{i,t}(\theta)=0$ a.s., $\forall i\in\mathcal{V}$.
\end{proof}
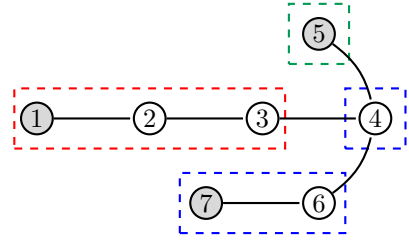
\begin{figure}[t]
\begin{center}
\begin{tikzpicture}
[->,shorten >=1pt,scale=.75,inner sep=1pt, minimum size=12pt, auto=center, node distance=3cm,
  thick, node/.style={circle, draw=black, thick},]
\tikzstyle{block1} = [rectangle, draw, fill=red!10, 
    text width=8em, text centered, rounded corners, minimum height=0.8cm, minimum width=1cm];
\node [circle, draw,fill=gray!30](n1) at (0,0)     (1)  {$1$};
\node [circle, draw](n2) at (2,0)     (2)  {$2$};
\node [circle, draw](n3) at (4,0)     (3)  {$3$};
\node [circle, draw](n4) at (6,0)     (4)  {$4$};
\node [circle, draw, fill=gray!30](n5) at (5,1.5)     (5)  {$5$};
\node [circle, draw](n6) at (5,-1.5)     (6)  {$6$};
\node [circle, draw, fill=gray!30](n7) at (3,-1.5)     (7)  {$7$};
\node (rect) at (2,0) () [draw=red, dashed, minimum width=3.6cm, minimum height=0.8cm] {};
\node (rect) at (4,-1.5) () [draw=blue, dashed, minimum width=2.2cm, minimum height=0.8cm] {};
\node (rect) at (6,0) () [draw=blue, dashed, minimum width=0.8cm, minimum height=0.8cm] {};
\node (rect) at (5,1.5) () [draw=ForestGreen, dashed, minimum width=0.8 cm, minimum height=0.8cm] {};
\draw [-, thick] (1) to (2);
\draw [-, thick] (2) to (3);
\draw [-, thick] (3) to (4);
\draw [-, thick] (7) to (6);
\draw [-,thick] (5) to [bend left=20] (4);
\draw [-,thick] (6) to [bend right=20] (4);
\end{tikzpicture}
\end{center}
\caption{The figure represents the network for the simulation example in Section \ref{sec:sim}. Based on the parameters of the model, Theorem \ref{thm:main} implies that the asymptotic rates of rejection of $\theta_2$ for the agents enclosed in the red, blue and green rectangles are dictated by the relative entropies of agents 1, 7 and 5, respectively, illustrating the agent-specific learning rate phenomenon.}
\label{fig:example}
\end{figure}
\begin{figure}[t]
\begin{center}
\begin{tabular}{cc}
\includegraphics[scale=0.235]{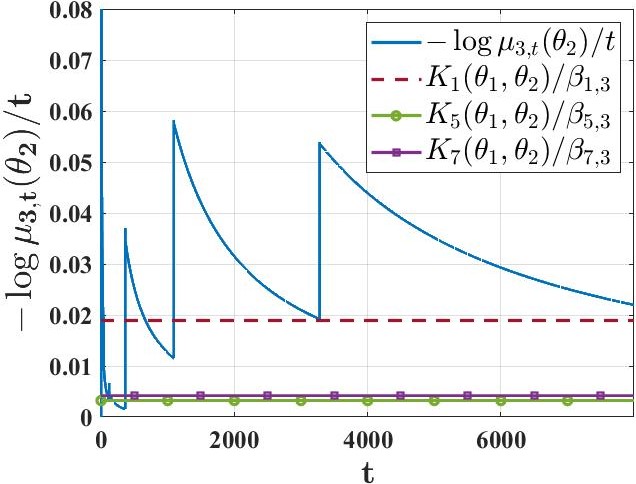}&\includegraphics[scale=0.235]{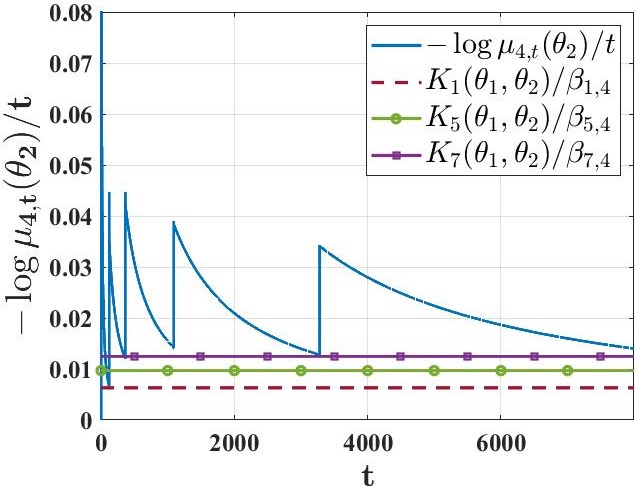}\\
\end{tabular}
\end{center}
\caption{The figure plots the instantaneous rates of decay of the beliefs of agents $3$ and $4$ on the false hypothesis $\theta_2$ (given by $-\log{\mu_{i,t}(\theta_2)}/{t}, i\in\{3,4\}$), for the model described in Section \ref{sec:sim}. The parameter $\beta_{i,j}=a^{(d(i,j)+1)}$ in the above plots represents the factor by which the signal strength of agent $i$ is attenuated at the location of agent $j$.}   
\label{fig:plots}
\end{figure}
\section{Simulation Example}
\label{sec:sim}
Consider a binary hypothesis testing scenario where $\Theta=\{\theta_1,\theta_2\}$, and $\theta_1$ is the true state of the world. The network of agents is depicted in Figure \ref{fig:example}. The signal space for every agent is identical, and given by $\mathcal{S}_i=\{1,2\}, \forall i\in\{1,\ldots,7\}$. The agent likelihood models satisfy: $l_i(1|\theta_1)=0.5,\forall i\in\{1,\dots,7\}, l_1(1|\theta_2)=0.9,  l_5(1|\theta_2)=0.7,l_7(1|\theta_2)=0.85$, and $l_i(1|\theta_2)=0.5, \forall i\in\{2,3,4,6\}.$  Thus, only agents $1, 5$ and $7$ can distinguish between $\theta_1$ and $\theta_2$, with their relative entropies satisfying $K_1(\theta_1,\theta_2)>K_7(\theta_1,\theta_2)>K_5(\theta_1,\theta_2)>0$ (all other agents have $K_i(\theta_1,\theta_2)=0)$. With  $a=3$, we have  $K_7(\theta_1,\theta_2)/{a^2}>K_5(\theta_1,\theta_2)/{a}>K_1(\theta_1,\theta_2)/{a^3}$. Figure \ref{fig:plots} plots the instantaneous rates of rejection of the false hypothesis $\theta_2$ for agents $3$ and $4$, resulting from our proposed algorithm. Based on Figures \ref{fig:example} and \ref{fig:plots}, a few key observations are: (i) each informative agent dominates the speed of learning of agents that are close to it in the network, (ii) the rate of rejection of the false hypothesis is indeed agent-specific, and (iii) the simulation results agree very closely with the theoretical lower bounds on the limiting rates of rejection in Theorem \ref{thm:main}. 

\section{The Impact of Information Allocation on Asymptotic Learning Rates}
\label{sec:infoalloc}
Theorem \ref{thm:main} indicates that the asymptotic learning rates of the agents are shaped by a non-trivial interplay between the relative entropies of their signal models and the structure of the network. In view of this fact, our next goal is to conduct a preliminary analysis of how information should be allocated to the agents in order to maximize appropriate performance metrics that are a function of the asymptotic learning rates. Our investigation is inspired by similar questions in \cite{jad2}; however, as we discuss next, our formulation differs considerably from \cite{jad2}. Specifically, unlike \cite{jad2}, our proposed learning rule leads to asymptotic learning rates that are agent-dependent when $a>1$ (as seen in Section \ref{sec:sim}). Consequently, the performance metrics that we seek to maximize differ from those in \cite{jad2}. As we shall soon see, while the eigenvector centrality plays a key role in shaping the speed of learning in \cite{jad2}, alternate network centrality measures become important when it comes to the belief dynamics generated by our rule. 

To make the above ideas precise, suppose we are given a strongly-connected communication graph $\mathcal{G}$, and a set of $n$ signal structures $\mathcal{L}=\{l_1,\ldots,l_n\}$, where each $l_i$ represents a family of parameterized marginals as defined in Section  \ref{sec:model}. By an allocation of signal structures to agents, we imply a bijection $\psi:\mathcal{L}\rightarrow\mathcal{V}$ between the elements of $\mathcal{L}$ and the elements of the vertex set of $\mathcal{G}$, namely $\mathcal{V}$. Let $\Psi$ represent the set of all possible bijections between the elements of $\mathcal{L}$ and $\mathcal{V}$. Our objective is to optimally pick $\psi\in\Psi$ so as to maximize the performance metrics that we define next. To this end, given a distinct pair of hypotheses $\theta_p,\theta_q\in\Theta$, recall from \eqref{eqn:asymprate} that based on our proposed learning rule,
\begin{equation}
    \rho^{\psi}_i(\theta_p,\theta_q)\triangleq\max_{v\in\mathcal{S}^{\psi}(\theta_p,\theta_q)}\frac{K^{\psi}_v(\theta_p,\theta_q)}{a^{(d(v,i)+1)}}
    \label{eqn:ratealloc}
\end{equation}
lower bounds the limiting rate at which agent $i$ rules out $\theta_q$ when $\theta_p$ is realized as the true state; the superscript $\psi$ reflects the dependence of the corresponding objects on the allocation policy $\psi.$ We now introduce two measures of the quality of learning that are specific to our setting:
\begin{equation}
\resizebox{1\hsize}{!}{$
\rho^{\psi}_{\textrm{avg}}\triangleq \min\limits_{\theta_p,\theta_q\in\Theta} \frac{1}{n}\sum_{i\in\mathcal{V}} \rho^{\psi}_i(\theta_p,\theta_q),
   \rho^{\psi}_{\textrm{min}}\triangleq \min\limits_{\theta_p,\theta_q\in\Theta} \min\limits_{i\in\mathcal{V}} \rho^{\psi}_i(\theta_p,\theta_q).$}
\end{equation}
While $\rho^{\psi}_{\textrm{avg}}$ captures the average rate of learning across the network, $\rho^{\psi}_{\textrm{min}}$ focuses on the agent that converges the slowest; given that any state in $\Theta$ can be realized, these metrics account for the pair of states that are the hardest to tell apart. We seek to maximize $\rho^{\psi}_{\textrm{avg}}$ and $\rho^{\psi}_{\textrm{min}}$ over the set of allocations $\Psi$. Our first result on this topic makes a connection to two popular network centrality measures, namely, the \textit{eccentricity centrality} and the \textit{decay centrality}, defined as follows. For a strongly-connected graph $\mathcal{G}$, the eccentricity centrality $\xi_i$\cite{hage}, and the decay centrality $\kappa_i(\delta)$ \cite{tsakas}, of an agent $i\in\mathcal{V}$ are given by
\begin{equation}
    \xi_i=\frac{1}{\max_{j\in\mathcal{V}\setminus\{i\}}d(i,j)}, \hspace{2mm}
    \kappa_i(\delta)=\sum_{j\in\mathcal{V}\setminus\{i\}}\delta^{d(i,j)},
    \label{eqn:centralities}
\end{equation}
where $0<\delta<1$ is the decay parameter.

The eccentricity centrality is a distance-based centrality measure that aims to find the `center' of a graph such that a process originating at the center minimizes the response time to any other agent. The decay centrality is also a closeness-based centrality measure where an agent is rewarded for being close to other agents, with agents at higher distances contributing less to the centrality as compared to those that are closer. We have the following result.
\begin{proposition} Let $\mathcal{G}$ be strongly-connected. Suppose $a>1$, and let there exist a signal structure $l_u\in\mathcal{L}$ such that the following is true for all  $\theta_p,\theta_q\in\Theta$,
\begin{equation}
    \frac{K_{l_u}(\theta_p,\theta_q)}{a^{\bar{d}(\mathcal{G})}} > K_{l_w}(\theta_p,\theta_q), \forall l_w\in\mathcal{L}\setminus\{l_u\}.\footnote{Here, the quantity  $K_{l_u}(\theta_p,\theta_q)$ should be interpreted differently from $K_{u}(\theta_p,\theta_q)$; whereas the former indicates a relative entropy associated with the signal structure $l_u$, the latter indicates a relative entropy associated with agent $u$ once it has been allocated a certain signal structure (which may not necessarily be $l_u$).}
    \label{eqn:condition}
\end{equation}
Then, (i) any allocation $\psi\in\Psi$ such that $\psi(l_u)\in\argmax_{i\in\mathcal{V}}\xi_i$ maximizes $\rho^{\psi}_{\textrm{min}}$, and (ii) any allocation $\psi\in\Psi$ such that $\psi(l_u)\in\argmax_{i\in\mathcal{V}}\kappa_i(\frac{1}{a})$ maximizes $\rho^{\psi}_{\textrm{avg}}$.
\label{prop:alloc}
\end{proposition}
\begin{proof}
For part (i), consider two allocations $\psi_1,\psi_2\in\Psi$ such that $\psi_1(l_u)=x_1\in\argmax_{i\in\mathcal{V}}\xi_i$, and $\psi_2(l_u)=x_2$. Based on condition \eqref{eqn:condition}, and \eqref{eqn:ratealloc}, it is easy to see that for any pair $\theta_p,\theta_q\in\Theta$, and for each $i\in\mathcal{V}$:
\begin{equation}
    \rho^{\psi_1}_i(\theta_p,\theta_q)=\frac{{K}^{\psi_1}_{x_1}(\theta_p,\theta_q)}{a^{(d(x_1,i)+1)}}, \rho^{\psi_2}_i(\theta_p,\theta_q)=\frac{{K}^{\psi_2}_{x_2}(\theta_p,\theta_q)}{a^{(d(x_2,i)+1)}}.
    \label{eqn:domination}
\end{equation}
Based on \eqref{eqn:centralities}, we then obtain:
\begin{equation}
\resizebox{1\hsize}{!}{$
\begin{aligned}
    \min_{i\in\mathcal{V}}\rho^{\psi_1}_i(\theta_p,\theta_q)-\min_{i\in\mathcal{V}}\rho^{\psi_2}_i(\theta_p,\theta_q)&=
    \frac{{K}^{\psi_1}_{x_1}(\theta_p,\theta_q)}{a^{(1/ \xi_{x_1}+1)}}- \frac{{K}^{\psi_2}_{x_2}(\theta_p,\theta_q)}{a^{(1/ \xi_{x_2}+1)}}\\
    &=\frac{K_{l_u}(\theta_p,\theta_q)}{a}\left(\frac{1}{a^{1/ \xi_{x_1}}}-\frac{1}{a^{1/ \xi_{x_2}}}\right)\\
    &\geq 0,
\end{aligned}
\label{eqn:allocineq}
$}
\end{equation}
where the second equality follows from the fact that the signal structure of agent $x_1$ under allocation $\psi_1$, and agent $x_2$ under allocation $\psi_2$, are each equal to $l_u$, and the last inequality follows by noting that $\xi_{x_1}\geq\xi_{x_2}$ based on the choice of agent $x_1$. The proof of part (i) then follows by noting that the inequality in \eqref{eqn:allocineq} holds for every pair $\theta_p,\theta_q\in\Theta.$ For part (ii), we proceed as in part (i) and compare two allocations $\psi_1,\psi_2\in\Psi$ such that $\psi_1(l_u)=x_1\in\argmax_{i\in\mathcal{V}}\kappa_i(\frac{1}{a})$, and $\psi_2(l_u)=x_2$. The equalities in \eqref{eqn:domination} hold once again, and combined with \eqref{eqn:centralities} lead to:
\begin{equation}
    \frac{1}{n}\sum_{i\in\mathcal{V}} \rho^{\psi_j}_i(\theta_p,\theta_q)=\frac{K^{\psi_j}_{x_j}(\theta_p,\theta_q)}{an}\left(1+\kappa_{x_j}\left(\frac{1}{a}\right)\right),
\end{equation}
where $j\in\{1,2\}$. The proof can be completed as in part (i) by noting that $\kappa_{x_1}(\frac{1}{a})\geq\kappa_{x_2}(\frac{1}{a})$.
\end{proof}

The intuition behind the above result is simple, and as follows. Suppose there exists a signal structure that is sufficiently stronger in its discriminatory power than the others  w.r.t. every pair of hypotheses. Then, the agent allocated such a structure will govern the rate of learning of every other agent in the network. To expedite learning, it thus makes sense to allocate such a dominant signal structure to the most central agent in the network (where the specific centrality measure depends on the performance metric).  
\begin{remark}
We point out that eccentricity centrality and decay centrality have been widely studied in the context of information spread over social and economic networks \cite{jalili,jackson,chatterjee}. For instance, while the former  bears connections to information cascades \cite{jalili}, the latter facilitates the selection of an ``implant" node that maximizes the diffusion of a certain product or idea over a network \cite{chatterjee}. Proposition \ref{prop:alloc} identifies conditions under which the above centrality measures have similar implications for the belief dynamics generated by our proposed learning rule.  
\end{remark}

While  Proposition \ref{prop:alloc} allows one to identify the optimal allocation strategy by simply computing the appropriate centrality measures, the scenario becomes much more complicated if no additional structure is imposed either on the network or on the agents' likelihood models. For such general cases, we provide a coarse upper bound on the suboptimality of any given allocation.

\begin{proposition}
Let $\mathcal{G}$ be strongly-connected.
 Suppose $a>1$, and let $\psi^{\star}_{\alpha}\in\Psi$ and $\psi^{\star}_{\beta}\in\Psi$ be allocations that maximize $\rho^{\psi}_{\textrm{min}}$ and $\rho^{\psi}_{\textrm{avg}}$, respectively.  Then, for any allocation $\psi\in\Psi$,
\begin{equation}
    \frac{\rho^{\psi^{\star}_{\alpha}}_{\textrm{min}}}{\rho^{\psi}_{\textrm{min}}} \leq a^{\bar{d}(\mathcal{G})}, \hspace{2mm} 
    \frac{\rho^{\psi^{\star}_{\beta}}_{\textrm{avg}}}{\rho^{\psi}_{\textrm{avg}}} \leq a^{\bar{d}(\mathcal{G})}.
    \label{eqn:deviation}
\end{equation}
\label{prop:deviation}
\end{proposition}
\begin{proof}
We only prove the second inequality in \eqref{eqn:deviation} since the first follows from similar arguments. Consider any $\psi\in\Psi$, and suppose the pair $(\theta_m,\theta_n)$ minimizes $\rho^{\psi}_{\textrm{avg}}$ for this allocation. The following inequality is then apparent from the definition of the quantities involved:
\begin{equation}
    \frac{\rho^{\psi^{\star}_{\beta}}_{\textrm{avg}}}{\rho^{\psi}_{\textrm{avg}}} \leq \frac{\sum_{i\in\mathcal{V}} \rho^{\psi^{\star}_{\beta}}_i(\theta_m,\theta_n)}{\sum_{i\in\mathcal{V}} \rho^{\psi}_i(\theta_m,\theta_n)}.
    \label{eqn:interimbound}
\end{equation}
Now fix an agent $i$, and suppose that under the allocation $\psi^{\star}_{\beta}$, the signal structure that governs the quantity $\rho^{\psi^{\star}_{\beta}}_i(\theta_m,\theta_n)$ (i.e., the structure that maximizes the right hand side of \eqref{eqn:ratealloc}) is $l_u$. Suppose $l_u$ is allocated to agents $v_1$ and $v_2$ under $\psi^{\star}_{\beta}$ and $\psi$, respectively. An inspection of \eqref{eqn:ratealloc} then reveals:
\begin{equation}
     \frac{\rho^{\psi^{\star}_{\beta}}_i(\theta_m,\theta_n)}{\rho^{\psi}_i(\theta_m,\theta_n)} \leq a^{d(v_2,i)-d(v_1,i)} \leq a^{\bar{d}(\mathcal{G})}.
\end{equation}
The above bound applies to every agent $i\in\mathcal{V}$, and hence, substituting it in \eqref{eqn:interimbound} leads to the desired result. 
\end{proof}
\section{Conclusion}
We developed and analyzed a simple time-triggered protocol for achieving communication-efficient non-Bayesian learning over a network. Unlike existing approaches, we allowed the inter-communication intervals to grow unbounded over time at an arbitrarily large (but finite) geometric rate $a\geq1$. We showed that despite such sparse communication, our approach enables each agent to learn the true state exponentially fast with probability 1. We then characterized the limiting error exponents of the agents as a function of the primitives of our model and the parameter $a$. For the special case when communication occurs at every time-step, i.e., when $a=1$, we proved that our approach yields strictly better asymptotic learning rates than those existing in the literature. Finally, for $a>1$, we studied the impact of signal allocations on the speed of learning. As future work, we plan to explore event-triggered rules for the problem under consideration, and investigate in more detail the aspect of information allocation initiated in Section \ref{sec:infoalloc}. 
\bibliographystyle{IEEEtran} 
\bibliography{refs}
\end{document}